\documentclass{article}

\usepackage{amsmath,amsfonts,amssymb,amsthm}
\usepackage{slashed}
\usepackage{xy}
\usepackage[breaklinks=true]{hyperref}
\usepackage{float}

\input xy
\xyoption{all}

\newcommand{\maps}{\colon}    
\newcommand{\R}{{\mathbb R}}  
\newcommand{\C}{{\mathbb C}}  
\renewcommand{\H}{{\mathbb H}}  
\renewcommand{\O}{{\mathbb O}}  
\newcommand{\K}{{\mathbb K}}  
\newcommand{\Z}{{\mathbb Z}}  

\renewcommand{\Re}{\mathrm{Re}} 
\renewcommand{\Im}{\mathrm{Im}} 
\newcommand{\Retr}{\Re \; \tr} 
\newcommand{\h}{\mathfrak{h}} 
\newcommand{\tr}{{\mathrm{tr}}} 

\newcommand{\U}{{\rm U}}    
\newcommand{\Spin}{{\rm Spin}}    

\newcommand{\so}{{\mathfrak{so}}}  
\newcommand{\gl}{{\mathfrak{gl}}}  
\newcommand{\g}{\mathfrak{g}}  
\newcommand{\T}{\mathcal{T}} 
\renewcommand{\b}{{\mathfrak{b}}}  
\newcommand{\siso}{\mathfrak{siso}} 
\newcommand{\sugra}{\mathfrak{sugra}} 
\newcommand{\brane}{\mathfrak{brane}} 
\newcommand{\superstring}{\mathfrak{superstring}} 

\newcommand{\End}{{\rm End}} 
\newcommand{\Hom}{{\rm Hom}} 
\newcommand{\Sym}{{\rm Sym}} 
\newcommand{\Cliff}{{\rm Cliff}}    
\newcommand{\ev}{{\rm ev}} 

\newcommand{\inclusion}{\hookrightarrow}
\newcommand{\iso}{\cong} 
\newcommand{\tensor}{\otimes} 

\newcommand{\fourth}{\frac{1}{4}} 

\newcommand{\V}{\mathcal{V}} 
\renewcommand{\S}{\mathcal{S}} 
\newcommand{\A}{\mathcal{A}} 
\newcommand{\B}{\mathcal{B}} 

\newcommand{\Psibar}{\overline{\Psi}} 
\newcommand{\chibar}{\overline{\chibar}} 

\newcommand{\define}[1]{{\bf #1}}

\newtheorem{thm}{Theorem}    
\newtheorem{cor}[thm]{Corollary}

\newtheorem{prop}[thm]{Proposition}

\theoremstyle{definition}

\newtheorem{defn}[thm]{Definition}

        \newcommand{\be}{\begin{equation}}
        \newcommand{\ee}{\end{equation}}
        \newcommand{\ba}{\begin{eqnarray}}
        \newcommand{\ea}{\end{eqnarray}}
        \newcommand{\ban}{\begin{eqnarray*}}
        \newcommand{\ean}{\end{eqnarray*}}
        \newcommand{\barr}{\begin{array}}
        \newcommand{\earr}{\end{array}}

\title{Division Algebras and Supersymmetry II} 

\author{John C.\ Baez and John Huerta \\
\\
Department of Mathematics \\
University of California \\
Riverside, CA 92521 USA 
}

\date{July 27, 2010} 

\begin{document}
\maketitle

\begin{abstract}
\noindent
Starting from the four normed division algebras---the real numbers,
complex numbers, quaternions and octonions---a systematic procedure
gives a 3-cocycle on the Poincar\'e Lie superalgebra in dimensions 3,
4, 6 and 10.  A related procedure gives a 4-cocycle on the Poincar\'e
Lie superalgebra in dimensions 4, 5, 7 and 11.  In general, an
$(n+1)$-cocycle on a Lie superalgebra yields a `Lie $n$-superalgebra':
that is, roughly speaking, an $n$-term chain complex equipped with a
bracket satisfying the axioms of a Lie superalgebra up to chain
homotopy.  We thus obtain Lie 2-superalgebras extending the Poincar\'e
superalgebra in dimensions 3, 4, 6, and 10, and Lie 3-superalgebras
extending the Poincar\'e superalgebra in dimensions 4, 5, 7 and 11.
As shown in Sati, Schreiber and Stasheff's work on higher gauge
theory, Lie 2-superalgebra connections describe the parallel transport
of strings, while Lie 3-superalgebra connections describe the parallel
transport of 2-branes.  Moreover, in the octonionic case, these
connections concisely summarize the fields appearing in 10- and
11-dimensional supergravity.
\end{abstract}

\vskip 2em

\section{Introduction} \label{sec:intro}

There is a deep connection between supersymmetry and the four normed
division algebras: the real numbers, complex numbers, 
quaternions and octonions.  This can be seen in
super-Yang--Mills theory, in superstring theory, and in theories of
supermembranes and supergravity.  Most simply, the connection is
visible from the fact that the normed division algebras have
dimensions 1, 2, 4 and 8, while classical superstring theories and
minimal super-Yang--Mills theories live in spacetimes of
dimension \emph{two} higher: 3, 4, 6 and 10.  The simplest
classical super-2-brane theories make sense in spacetimes of dimensions 
\emph{three} higher: 4, 5, 7 and 11.  
Classical supergravity makes sense in all of these dimensions, but the
octonionic cases are the most important: in 10 dimensions supergravity
is a low-energy limit of superstring theory, while in 11 dimensions it 
is believed to be a low-energy limit of `M-theory', which incorporates 
the 2-brane.

These numerical relationships are far from coincidental.  They arise
because we can use the normed division algebras to construct the
spacetimes in question, as well as their associated spinors.  A
certain spinor identity that holds in dimensions 3, 4, 6 and
10 is an easy consequence of this construction, as is a related
identity that holds in dimensions 4, 5, 7 and 11.  These identities
are fundamental to the physical theories just listed.

\vfill \eject

In a bit more detail, suppose $\K$ is a normed division algebra of 
dimension $n$.   There are just four examples:
\begin{itemize}
\item the real numbers $\R$ ($n = 1$),
\item the complex numbers $\C$ ($n = 2$),
\item the quaternions $\H$ ($n = 4$),
\item the octonions $\O$ ($n = 8$).
\end{itemize}
Then we can identify vectors in $(n+2)$-dimensional Minkowski
spacetime with $2 \times 2$ hermitian matrices having entries in $\K$.
Similarly, we can identify spinors with elements of $\K^2$.  Matrix
multiplication then gives a way for vectors to act on spinors.  There
is also an operation that takes two spinors $\psi$ and $\phi$ and
forms a vector $\psi \cdot \phi$.  Using elementary properties of
normed division algebras, we can prove that
\[               (\psi \cdot \psi) \psi = 0 . \]
Following Schray \cite{Schray}, we call this identity the `3-$\psi$'s rule'. This identity is an example of a `Fierz identity'---roughly, 
an identity that allows one to reorder multilinear expressions made 
of spinors. This can be made more visible in the 3-$\psi$'s rule if 
we polarize the above cubic form to extract a genuinely trilinear expression:
\[               (\psi \cdot \phi) \chi + (\phi \cdot \chi) \psi + (\chi \cdot \psi) \phi = 0 . \]

In fact, the 3-$\psi$'s rule holds \emph{only} when Minkowski
spacetime has dimension 3, 4, 6 or 10.  Moreover, it is crucial for
super-Yang--Mills theory and superstring theory in these dimensions.
In minimal super-Yang--Mills theory, we need the 3-$\psi$'s rule to
check that the Lagrangian is supersymmetric, thanks to an argument
reviewed in our previous paper \cite{BaezHuerta}.  In superstring
theory, we need it to check the supersymmetry of the Green--Schwarz
Lagrangian \cite{GreenSchwarz, GreenSchwarzWitten}.  But the
3-$\psi$'s rule also has a deeper significance, which we study here.

This deeper story involves not only the 3-$\psi$'s rule but also
the `4-$\Psi$'s rule', a closely related Fierz identity required 
for super-2-brane theories in dimensions 4, 5, 7 and 11.  To help the
reader see the forest for the trees, we present a rough summary of 
this story in the form of a recipe:

\begin{enumerate}

	\item Spinor identities that come from division algebras are
		\emph{cocycle conditions}.

	\item The corresponding cocycles allow us to extend the Poincar\'e
		Lie superalgebra to a higher structure, a \emph{Lie
		$n$-superalgebra}.

	\item Connections valued in these Lie
		$n$-superalgebras describe the \emph{field content} of
		superstring and super-2-brane theories.

\end{enumerate}

To begin our story in dimensions 3, 4, 6 and 10, let us first introduce some
suggestive terminology: despite our notation, we shall call $\psi \cdot \phi$
the \define{bracket of spinors}. This is because this function is symmetric,
and it defines a Lie superalgebra structure on the supervector space
\[ \T = V \oplus S \]
where the even subspace $V$ is the vector representation of
$\Spin(n+1,1)$, while the odd subspace $S$ is a certain spinor
representation.  This Lie superalgebra is called the
\define{supertranslation algebra}.

There is a cohomology theory for Lie superalgebras, sometimes called
Chevalley--Eilenberg cohomology.  The cohomology of $\T$ will play a
central role in what follows.  Why?  First, because the 3-$\psi$'s is
really a cocycle condition, for a 3-cocycle $\alpha$ on $\T$ which eats 
two spinors and a vector and produces a number as follows:
\[ \alpha(\psi, \phi, A) = \langle \psi, A \phi \rangle. \]
Here, $\langle -,- \rangle$ is a pairing between spinors.  
Since this 3-cocycle is Lorentz-invariant, it extends to a cocycle
on the Poincar\'e superalgebra
\[    \siso(n+1,1) \iso \so(n+1,1) \ltimes \T  .\]
In fact, we obtain a nonzero element of the third cohomology of
the Poincar\'e superalgebra this way.

Just as 2-cocycles on a Lie superalgebra give ways of 
extending it to larger Lie superalgebras, 3-cocycles give
extensions to \emph{Lie 2-superalgebras}.  To understand this, we
need to know a bit about $L_\infty$-algebras \cite{MSS,SS}.  An
$L_\infty$-algebra is a chain complex equipped with a structure like
that of a Lie algebra, but where the laws hold only `up to $d$ of
something'.  A Lie $n$-algebra is an $L_\infty$-algebra in which only
the first $n$ terms are nonzero.  All these ideas also have `super'
versions.  In particular, we can use the 3-cocycle $\alpha$ to extend
$\siso(n+1,1)$ to a Lie 2-superalgebra of the following form:
\[ \xymatrix{ \siso(n+1, 1) & \R \ar[l]_<<<<<d }. \]
We call this the `superstring Lie 2-superalgebra', and denote it as 
$\superstring(n+1,1)$.   

The superstring Lie 2-superalgebra is an extension of $\siso(n+1,1)$ by 
$\b\R$, the Lie 2-algebra with $\R$ in degree 1 and everything else
trivial.  By `extension', we mean that there is a short exact sequence
of Lie 2-superalgebras:
\[         0 \to b\R \to \superstring(n+1,1) \to \siso(n+1,1) \to 0 . \]
To see precisely what this means, let us expand it a bit.  Lie
2-superalgebras are 2-term chain complexes, and writing these
vertically, our short exact sequence looks like this:
\[ 
\def\objectstyle{\scriptstyle}
\xymatrix { 
            0 \ar[r] & \R \ar[r] \ar[d]_d & \R \ar[r] \ar[d]_d    & 0 \ar[r] \ar[d]_d     & 0 \\
            0 \ar[r] & 0 \ar[r]         & \siso(n+1,1) \ar[r] & \siso(n+1,1) \ar[r] & 0 \\
} \]
In the middle, we see $\superstring(n+1,1)$.  This Lie 2-superalgebra
is built from two pieces: $\siso(n+1,1)$ in degree $0$ and $\R$ in
degree $1$.  But since the cocycle $\alpha$ is nontrivial, these two
pieces still interact in a nontrivial way.  Namely, the Jacobi
identity for three 0-chains holds only up to $d$ of a 1-chain.  So,
besides its Lie bracket, the Lie 2-superalgebra $\superstring(n+1,1)$
also involves a map that takes three 0-chains and gives a 1-chain.
This map is just $\alpha$.

What is the superstring Lie 2-algebra good for?  The answer lies in a
feature of string theory called the `Kalb--Ramond field', or `$B$
field'.  The $B$ field couples to strings just as the $A$ field in
electromagnetism couples to charged particles.  The $A$ field is
described locally by a 1-form, so we can integrate it over a
particle's worldline to get the interaction term in the Lagrangian for
a charged particle.  Similarly, the $B$ field is described locally by
a 2-form, which we can integrate over the worldsheet of a string.

Gauge theory has taught us that the $A$ field has a beautiful
geometric meaning: it is a connection on a $\U(1)$ bundle over
spacetime.  What is the corresponding meaning of the $B$ field?  It
can be seen as a connection on a `$\U(1)$ gerbe': a gadget like a
$\U(1)$ bundle, but suitable for describing strings instead of point
particles.  Locally, connections on $\U(1)$ gerbes can be identified
with 2-forms.  But globally, they cannot.  The idea that the $B$ field
is a $\U(1)$ gerbe connection is implicit in work going back at least
to the 1986 paper by Gawedzki \cite{Gawedzki}.  More recently, Freed
and Witten \cite{FreedWitten} showed that the subtle difference
between 2-forms and connections on $\U(1)$ gerbes is actually crucial
for understanding anomaly cancellation.  In fact, these authors used
the language of `Deligne cohomology' rather than gerbes.  Later work
made the role of gerbes explicit: see for example Carey, Johnson and
Murray \cite{CareyJohnsonMurray}, and also Gawedzki and Reis
\cite{GawedzkiReis}.

More recently still, work on higher gauge theory has revealed that the
$B$ field can be viewed as part of a larger package.  Just as gauge
theory uses Lie groups, Lie algebras, and connections on bundles to to
describe the parallel transport of point particles, higher gauge
theory generalizes all these concepts to describe parallel transport
of extended objects such strings and membranes
\cite{BaezHuerta:invitation, BaezSchreiber}.  In particular, Schreiber, Sati and
Stasheff \cite{SSS} have developed a theory of `$n$-connections'
suitable for describing parallel transport of objects with
$n$-dimensonal worldvolumes.  In their theory, the Lie algebra of the
gauge roup is replaced by an Lie $n$-algebra---or in the
supersymmetric context, a Lie $n$-superalgebra.  Applying their ideas
to $\superstring(n+1,1)$, we get a 2-connection which can be described
locally using the following fields:
\vskip 1em
\begin{center}
\begin{tabular}{cc}
	\hline
	$\superstring(n+1,1)$ & Connection component \\
	\hline
	$\R$                  & $\R$-valued 2-form \\
	$\downarrow$          & \\
	$\siso(n+1,1)$        & $\siso(n+1,1)$-valued 1-form \\
	\hline
\end{tabular}
\end{center}
\vskip 1em
The $\siso(n+1,1)$-valued 1-form consists of three fields which help
define the background geometry on which a superstring propagates: the
Levi-Civita connection $A$, the vielbein $e$, and the gravitino
$\psi$.  But the $\R$-valued 2-form is equally important in the
description of this background geometry: it is the $B$ field!

Next let us extend these ideas to Minkowski spacetimes one dimension 
higher: dimensions 4, 5, 7 and 11.  In this case a certain subspace of 
$4 \times 4$ matrices with entries in $\K$ will form the vector representation of $\Spin(n+2, 1)$, while $\K^4$ will form a spinor 
representation.  As before, there is a `bracket' operation that takes two spinors $\Psi$ and $\Phi$ and gives a vector $\Psi \cdot \Phi$.  As 
before, there is an action of vectors on spinors.  This time the 3-
$\psi$'s rule no longer holds:
\[               (\Psi \cdot \Psi) \Psi \neq 0 . \]
However, we show that
\[ \Psi \cdot ((\Psi \cdot \Psi) \Psi) = 0 . \]
We call this the `4-$\Psi$'s rule'.  This identity plays a basic role 
for the super-2-brane, and related theories of supergravity.

Once again, the bracket of spinors defines a Lie superalgebra structure 
on the supervector space
\[ \T = \V \oplus \S \]
where now $\V$ is the vector representation of $\Spin(n+2,1)$, while
$\S$ is a certain spinor representation of this group.  Once again, the
cohomology of $\T$ plays a key role.  The 4-$\Psi$'s rule is a cocycle
condition---but this time for a 4-cocycle $\beta$ which eats two
spinors and two vectors and produces a number as follows:
\[ \beta(\Psi, \Phi, \A, \B) = \langle \Psi, (\A \wedge \B) \Phi \rangle. \]
Here, $\langle -, - \rangle$ denotes the inner product of two spinors, 
and the bivector $\A \wedge \B$ acts on $\Phi$ via the usual Clifford 
action. Since $\beta$ is Lorentz-invariant, we shall see that it 
extends to a 4-cocycle on the Poincar\'e superalgebra $\siso(n+2,1)$.

We can use $\beta$ to extend the Poincar\'e superalgebra to a Lie
3-superalgebra of the following form:
\[ \xymatrix{ \siso(n+2, 1) & 0 \ar[l]_<<<<<d & \R \ar[l]_d }. \]
We call this the `2-brane Lie 3-superalgebra', and denote it as 
$\mbox{2-}\brane(n+1,1)$.   It is an extension of $\siso(n+2,1)$ by 
$\b^2\R$, the Lie 3-algebra with $\R$ in degree 2, and everything
else trivial.  In other words, there is a short exact sequence:
\[ 0 \to \b^2\R \to \mbox{2-}\brane(n+2,1) \to \siso(n+2,1) \to 0 \]
Again, to see what this means, let us expand it a bit.  Lie
3-superalgebras are 3-term chain complexes.  Writing out each of these 
vertically, our short exact sequence looks like this:
\[ 
\def\objectstyle{\scriptstyle}
\xymatrix { 
            0 \ar[r]  & \R \ar[r] \ar[d]_d & \R \ar[r] \ar[d]_d    & 0 \ar[r] \ar[d]_d     & 0  \\
            0 \ar[r]  & 0 \ar[r] \ar[d]_d  & 0 \ar[r] \ar[d]_d     & 0 \ar[r] \ar[d]_d     & 0  \\
            0 \ar[r]        & 0 \ar[r]         & \siso(n+2,1) \ar[r] & \siso(n+2,1) \ar[r] & 0 \\
} \]
In the middle, we see 2-$\brane(n+2,1)$.

The most interesting Lie 3-algebra of this type, 2-$\brane(10,1)$, 
plays an important role in 11-dimensional supergravity.  This idea goes back to the work of Castellani, D'Auria and Fr\'e \cite
{TheCube, DAuriaFre}.  These authors derived the field content of 11-
dimensional supergravity starting from a differential graded commutative algebra.  Later, Sati, Schreiber and Stasheff 
\cite{SSS} explained that these fields can be reinterpreted as a
3-connection valued in a Lie 3-algebra which they called 
`$\sugra(10,1)$'.  This is the Lie 3-algebra we are calling 
2-$\brane(10,1)$.  Our message here is that the all-important 
cocycle needed to construct this Lie 3-algebra arises naturally from 
the octonions, and has analogues for the other normed division 
algebras.

If we follow these authors and consider a 3-connection valued in 
2-$\brane(10,1)$, we find it can be described locally by these fields:
\vskip 1em
\begin{center}
\begin{tabular}{cc}
	\hline
	2-$\brane(n+2,1)$ & Connection component \\
	\hline
	$\R$            & $\R$-valued 3-form \\
	$\downarrow$    & \\
	$0$             & \\
	$\downarrow$    & \\
	$\siso(n+2,1)$  & $\siso(n+2,1)$-valued 1-form \\
	\hline
\end{tabular}
\end{center}
\vskip 1em
Again, a $\siso(n+2,1)$-valued 1-form contains familiar fields: the
Levi-Civita connection, the vielbein, and the gravitino.  But now we
also see a 3-form, called the $C$ field.  This is again something we
might expect on physical grounds, at least in dimension 11.  While
the case is less clear than in string theory, it seems that for the
quantum theory of a 2-brane to be consistent, it must propagate in
a background obeying the equations of 11-dimensional supergravity, 
in which the $C$ field naturally shows up \cite{Tanii}.  The work 
of Diaconescu, Freed, and Moore \cite{DiaconescuFreedMoore}, as well
as that of Aschieri and Jurco \cite{AschieriJurco}, is also relevant here.

Finally, we mention another use for the cocycles $\alpha$ and $\beta$.
These cocycles are also used to build Wess--Zumino--Witten terms for
superstrings and 2-branes. For example, in the case of the string, one
can extend the string's worldsheet to be the boundary of a
three-dimensional manifold, and then integrate $\alpha$ over this
manifold.  This provides an additional term for the action of the
superstring, a term that is required to give the action Siegel
symmetry, balancing the number of bosonic and fermionic degrees of
freedom. For the 2-brane, the Wess--Zumino--Witten term is constructed
in complete analogy---we just `add one' to all the dimensions in
sight \cite{AchucarroEvans,Duff}.

Indeed, the network of relationships between supergravity, string and
2-brane theories, and cocycles constructed using normed division
algebras is extremely tight.  The Siegel symmetry of the string or
2-brane action constrains the background of the theory to be that of
supergravity, at least in dimensions 10 and 11 \cite{Tanii}, and
without the WZW terms, there would be no Siegel symmetry.  The WZW
terms rely on the cocycles $\alpha$ and $\beta$.  These cocycles also
give rise to the Lie 2- and 3-superalgebras $\superstring(9,1)$ and
2-$\brane(10,1)$.  And these, in turn, describe the field content of
supergravity in these dimensions!

As further grist for this mill, WZW terms can also be viewed in the
context of higher gauge theory.  In string theory, the WZW term is the holonomy of a connection on a $\U(1)$ gerbe \cite{GawedzkiReis}.  Presumably the WZW term in a 2-brane theory is the holonomy of a connection on a $\U(1)$ 2-gerbe \cite{Stevenson}.  This is a 
tantalizing clue that we are at the beginning of a larger but 
ultimately simpler story.

In what follows we focus on the mathematics of constructing Lie
$n$-superalgebras from normed division algebras, rather than the
applications to physics that we have just described.  We begin
with a quick review of normed division algebras in Section 
\ref{sec:divalg}.  In Section \ref{sec:n+2} we recall how an
$n$-dimensional normed division algebra can be used to describe vectors 
and spinors in $(n+2)$-dimensional spacetime, and how this description
yields the 3-$\psi$'s rule.  All this material is treated in more detail 
in our previous paper \cite{BaezHuerta}.  In Section \ref{sec:n+3}
we build on this work and use normed division algebras to describe 
vectors and spinors in $(n+2)$-dimensional spacetime.  In
Section \ref{sec:4psi} we use this description to prove the 
the 4-$\Psi$'s rule.  In Section \ref{sec:cohomology} we describe
the cohomology of Lie superalgebras, and show that the 3-$\psi$'s
rule and 4-$\Psi$'s rule yield nontrivial 3-cocycles and 4-cocycles on 
supertranslation algebras.  In Section \ref{sec:Linfinity} we 
describe Lie $n$-superalgebras, and prove that an $(n+1)$-cocycle on 
a Lie superalgebra gives a way to extend it to a Lie $n$-superalgebra.
We thus obtain Lie 2-superalgebras and Lie 3-superalgebras extending
supertranslation algebras.  In Section \ref{sec:algebras} we conclude
by constructing the superstring Lie 2-algebras and the 2-brane 
Lie 3-algebras.

\section{Division Algebras} \label{sec:divalg}

We begin with a lightning review of normed division algebras.
A \define{normed division algebra} $\K$ is a (finite-dimensional,
possibly nonassociative) real algebra equipped with a multiplicative
unit 1 and a norm $| \cdot |$ satisfying:
\[ |ab| = |a| |b|  \]
for all $a, b \in \K$.  Note this implies that $\K$ has no zero
divisors.  We will freely identify $\R 1 \subseteq \K$ with $\R$.  By
a classic theorem of Hurwitz~\cite{Hurwitz}, there are only four
normed division algebras: the real numbers, $\R$, the complex numbers,
$\C$, the quaternions, $\H$, and the octonions, $\O$.  These algebras
have dimension 1, 2, 4, and 8.

In all four cases, the norm can be defined using conjugation. Every normed
division algebra has a \define{conjugation} operator---a linear
operator $* \maps \K \to \K$ satisfying
\[ a^{**} = a, \quad (ab)^* = b^* a^* \]
for all $a,b \in \K$.   Conjugation lets us decompose each element of
$\K$ into real and imaginary parts, as follows:
\[ \Re(a) = \frac{a + a^*}{2}, \quad \Im(a) = \frac{a - a^*}{2}. \]
Conjugating changes the sign of the imaginary part and leaves the real part
fixed. We can write the norm as
\[ |a| = \sqrt{a a^*} = \sqrt{a^* a}. \]
This norm can be polarized to give an inner product on $\K$:
\[ (a, b) = \Re(a b^*) = \Re(a^* b). \]

The algebras $\R$, $\C$ and $\H$ are associative. The octonions $\O$
are not.  Instead, they are \define{alternative}: the subalgebra
generated by any two octonions is associative. Note that $\R$, $\C$
and $\H$, being associative, are also trivially alternative. By a
theorem of Artin \cite{Schafer}, this is equivalent to the fact that the
\define{associator}
\[ [a,b,c] = (ab)c - a(bc) \]
is completely antisymmetric in its three arguments.

For any square matrix $A$ with entries in $\K$, we define its
\define{trace} $\tr(A)$ to be the sum of its diagonal entries. This
trace lacks the usual cyclic property, because $\K$ is noncommutative,
so in general $\tr(AB) \neq \tr(BA)$.  Luckily, taking the real part
restores this property:

\begin{prop}
\label{prop:realtrace}
Let $A$, $B$, and $C$ be $k \times \ell$, $\ell \times m$ and
$m \times k$ matrices with entries in $\K$. Then
\[ \Retr((AB)C) = \Retr(A(BC)) \]
and this quantity is invariant under cyclic permutations of $A$, $B$,
and $C$.  We call this quantity the \define{real trace} $\Retr(ABC)$.
\end{prop}

\begin{proof}
	This relies heavily on the alternativity of normed division algebras.
	See Proposition 4 of our previous paper \cite{BaezHuerta}.
\end{proof}

\section{Spacetime Geometry in $n+2$ Dimensions} \label{sec:n+2}

In this section we recall the relation between a normed division algebra $\K$
of dimension $n$ and Lorentzian geometry in $n+2$ dimensions.  Most of the
material here is well-known \cite{Baez:Octonions, ChungSudbery, KugoTownsend,  
ManogueSudbery, Sudbery}, and we follow the treatment in Section 3 of our
previous paper \cite{BaezHuerta}, much of which we learned from the papers of
Manogue and Schray \cite{Schray, SchrayManogue}.   The key facts are that one
can describe vectors in $(n+2)$-dimensional Minkowski spacetime as $2 \times 2$
hermitian matrices with entries in $\K$, and spinors as elements of $\K^2$.  In
fact there are two representations of $\Spin(n+1,1)$ on $\K^2$, which we call
$S_+$ and $S_-$.  The nature of these representations depends on $\K$:
\begin{itemize}
    \item When $\K = \R$, $S_+ \iso S_-$ is the Majorana spinor
        representation of $\Spin(2,1)$.
    \item When $\K = \C$, $S_+ \iso S_-$ is the Majorana spinor
        representation of $\Spin(3,1)$. 
    \item When $\K = \H$, $S_+$ and $S_-$ are the Weyl spinor
        representations of $\Spin(5,1)$.
    \item When $\K = \O$, $S_+$ and $S_-$ are the Majorana--Weyl
        spinor representations of $\Spin(9,1)$.
\end{itemize} 
As usual, these spinor representations are also
representations of the even part of the relevant Clifford algebras:
\vskip 1em
\begin{center}
\renewcommand{\arraystretch}{1.4}
\begin{tabular}{|lcl|}
	\hline
	\multicolumn{3}{|c|}{\textbf{Even parts of Clifford algebras}} \\
	\hline
	$\Cliff_{\ev}(2,1)$ & $\iso$ & $\R[2]$                \\
	$\Cliff_{\ev}(3,1)$ & $\iso$ & $\C[2]$                \\
	$\Cliff_{\ev}(5,1)$ & $\iso$ & $\H[2] \oplus \H[2]$   \\
	$\Cliff_{\ev}(9,1)$ & $\iso$ & $\R[16] \oplus \R[16]$ \\
	\hline
\end{tabular}
\renewcommand{\arraystretch}{1}
\end{center}
Here we see $\R^2$, $\C^2$, $\H^2$ and $\O^2$ showing up as irreducible
representations of these algebras, albeit with $\O^2$
masquerading as $\R^{16}$. The first two algebras have a unique irreducible
representation.  The last two both have two irreducible representations,
which correspond to left-handed and right-handed spinors.

Our discussion so far has emphasized the differences between the 4
cases.  But the wonderful thing about normed division algebras is that
they allow a unified approach that treats all four cases
simultaneously!  They also give simple formulas for the basic
intertwining operators involving vectors, spinors and scalars.

To begin, let $\K[m]$ denote the space of $m \times m$ matrices with entries in $\K$.
Given $A \in \K[m]$, define its \define{hermitian adjoint} $A^\dagger$ to be
its conjugate transpose:
\[ A^\dagger = (A^*)^T. \]
We say such a matrix is \define{hermitian} if $A = A^\dagger$. Now take the $2
\times 2$ hermitian matrices:
\[ 
\h_2(\K) = \left\{ 
\left( 
\begin{array}{c c} 
	t + x & y     \\
	y^*   & t - x \\
\end{array}
\right) 
\; : \;
t, x \in \R, \; y \in \K
\right\}.
\]
This is an $(n + 2)$-dimensional real vector space. Moreover, the
usual formula for the determinant of a matrix gives the Minkowski norm
on this vector space:
\[ 
-\det 
\left( 
\begin{array}{c c} 
	t + x & y     \\
	y^*   & t - x \\
\end{array}
\right) 
= - t^2 + x^2 + |y|^2.
\]
We insert a minus sign to obtain the signature $(n+1, 1)$. Note this
formula is unambiguous even if $\K$ is noncommutative or nonassociative.   

It follows that the double cover of the Lorentz group, $\Spin(n + 1,
1)$, acts on $\h_2(\K)$ via determinant-preserving linear
transformations.  Since this is the `vector' representation, we will
often call $\h_2(\K)$ simply $V$.  The Minkowski metric
\[           g \maps V \otimes V \to \R  \]
is given by
\[           g(A,A) = -\det(A)  .\]
There is also a nice formula for the inner product of two different
vectors.  This involves the \define{trace reversal} of $A \in \h_2(\K)$,
defined by
\[ \tilde{A} = A - (\tr A) 1. \]
Note we indeed have $\tr(\tilde{A}) = -\tr(A)$.

\begin{prop}
\label{prop:metric}
For any vectors $A,B \in V = \h_2(K)$, we have
\[       A \tilde{A} = \tilde{A} A = - \det(A) 1 \]
and
\[  \frac{1}{2} \Retr(A \tilde{B}) = \frac{1}{2} \Retr(\tilde{A} B) = g(A,B)
\]
\end{prop}

Next we consider spinors.  As real vector spaces, the spinor
representations $S_+$ and $S_-$ are both just $\K^2$.
However, they differ as representations of $\Spin(n+1, 1)$. To
construct these representations, we begin by defining ways for vectors
to act on spinors:
\[ \begin{array}{cccl} 
	\gamma \maps & V \tensor S_+ & \to     & S_- \\
	             & A \tensor \psi     & \mapsto & A\psi.
   \end{array} \]
and 
\[ \begin{array}{cccl} 
	\tilde{\gamma} \maps & V \tensor S_- & \to     & S_+ \\
	                     & A \tensor \psi     & \mapsto & \tilde{A}\psi.
\end{array} \]
We can also think of these as maps that send elements of $V$ to linear 
operators:
\[ \begin{array}{cccl}
	\gamma \maps         & V & \to & \Hom(S_+, S_-), \\
	\tilde{\gamma} \maps & V & \to & \Hom(S_-, S_+).
\end{array} \]

Since vectors act on elements of $S_+$ to give elements of $S_-$
and vice versa, they map the space $S_+ \oplus S_-$ to itself.
This gives rise to an action of the Clifford algebra 
$\Cliff(V)$ on $S_+ \oplus S_-$:

\begin{prop} \label{prop:n2cliff}
The vectors $V = \h_2(\K)$ act on the spinors $S_+ \oplus S_- = \K^2
\oplus \K^2$ via the map
\[ \Gamma \maps  V \to      \End(S_+ \oplus S_-) \]
given by 
\[   \Gamma(A)(\psi, \,\phi) = (\widetilde{A} \phi, \, A \psi)  .\]
Furthermore, $\Gamma(A)$ satisfies the Clifford algebra relation:
\[ \Gamma(A)^2 = g(A,A) 1 \]
and so extends to a homomorphism $\Gamma \maps \Cliff(V) \to 
\End(S_+ \oplus S_-)$, i.e.\ a representation of the Clifford algebra 
$\Cliff(V)$ on $S_+ \oplus S_-$.
\end{prop}

As explained in our previous paper \cite{BaezHuerta}, the spaces $S_+$, $S_-$
and $V$ are representations of the spin group $\Spin(n+1,1)$. Moreover:

\begin{prop} 
The maps 
\[ \begin{array}{cccl}
	\gamma \maps & V \tensor S_+ & \to     & S_- \\
                     & A \tensor \psi     & \mapsto & A \psi
\end{array} \]
and
\[ \begin{array}{cccl}
	\tilde{\gamma} \maps & V \tensor S_- & \to     & S_+ \\
	                     & A \tensor \psi    & \mapsto & \tilde{A} \psi
\end{array} \]
are equivariant with respect to the action of $\Spin(n+1, 1)$.
\end{prop}

\begin{prop}
The pairing
\[ \begin{array}{cccl}
\langle -, - \rangle \maps & S_+ \tensor S_- & \to     & \R \\
 & \psi \tensor \phi & \mapsto & \Re(\psi^\dagger \phi)
\end{array} \]
is invariant under the action of $\Spin(n+1,1)$.
\end{prop}

With this pairing in hand, there is a manifestly equivariant way to
turn a pair of spinors into a vector.  Given 
$\psi, \phi \in S_+$, there is a unique vector $\psi \cdot \phi$ 
whose inner product with any vector $A$ is given by
\[ g(\psi \cdot \phi, A) = \langle \psi, \gamma(A) \phi \rangle .\]
Similarly, given $\psi, \phi \in S_-$, we define 
$\psi \cdot \phi \in V$ by demanding
\[ g(\psi \cdot \phi, A) = \langle \tilde{\gamma}(A) \psi, \phi \rangle \]
for all $A \in V$.  This gives us maps
\[ S_\pm \tensor S_\pm \to V \]
which are manifestly equivariant. In fact:

\begin{prop}
The maps $\cdot \, \maps S_\pm \tensor S_\pm \to V$ are given by:
\[ \begin{array}{cccl}
	\cdot \, \maps & S_+ \tensor S_+ & \to     & V \\
 & \psi \tensor \phi   & \mapsto & 
\widetilde{\psi \phi^\dagger + \phi \psi^\dagger}
\end{array} \]
\[ \begin{array}{cccl}
	\cdot \, \maps & S_- \tensor S_- & \to     & V \\
& \psi \tensor \phi  & \mapsto & 
\psi \phi^\dagger + \phi \psi^\dagger.
\end{array} \]
These maps are equivariant with respect to the action of $\Spin(n+1, 1)$.
\end{prop}

Theorem 10 of our previous paper stated the fundamental identity which allows
supersymmetry in dimensions 3, 4, 6 and 10: the `3-$\psi$'s rule'. 
Our proof was based on an argument in the appendix of 
a paper by Dray, Janesky and Manogue \cite{DrayJaneskyManogue}:

\begin{thm}  
\label{thm:fundamental_identity}
Suppose $\psi \in S_+$.  Then $(\psi \cdot \psi) \psi = 0$. Similarly, if 
$\phi \in S_-$, then $(\widetilde{\phi \cdot \phi}) \phi = 0$.
\end{thm}

%

\section{Spacetime Geometry in $n+3$ Dimensions} \label{sec:n+3}

In the last section we recalled how to describe spinors and vectors in
$(n+2)$-dimensional Minkowski spacetime using a division algebra $\K$
of dimension $n$.  Here we show how to boost this up one dimension,
and give a division algebra description of vectors and spinors in
$(n+3)$-dimensional Minkowski spacetime.

We shall see that vectors in $(n+3)$-dimensional Minkowski spacetime can
be identified with $4 \times 4$ $\K$-valued matrices of 
this particular form:
\[ \left( 
\begin{array}{cc} 
	a & \tilde{A} \\
	A & -a
\end{array}
\right)
\]
where $a$ is a real multiple of the $2 \times 2$ identity matrix
and $A$ is a $2 \times 2$ hermitian matrix with entries in $\K$.
Moreover, $\Spin(n+2, 1)$ has a representation on $\K^4$,
which we call $\S$.  Depending on $\K$, this gives the following
types of spinors: 
\begin{itemize}
    \item When $\K = \R$, $\S$ is the Majorana spinor representation of
	    $\Spin(3,1)$.
    \item When $\K = \C$, $\S$ is the Dirac spinor representation of
	    $\Spin(4,1)$. 
    \item When $\K = \H$, $\S$ is the Dirac spinor representation of
	    $\Spin(6,1)$.
    \item When $\K = \O$, $\S$ is the Majorana spinor representation of
	    $\Spin(10,1)$.
\end{itemize} 
Again, these spinor representations are also
representations of the even part of the relevant Clifford algebra:
\vskip 1em
\begin{center}
\renewcommand{\arraystretch}{1.4}
\begin{tabular}{|lcl|}
	\hline
	\multicolumn{3}{|c|}{\textbf{Even parts of Clifford algebras}} \\
	\hline
	$\Cliff_{\ev}(3,1)$  & $\iso$ & $\C[2]$                 \\
	$\Cliff_{\ev}(4,1)$  & $\iso$ & $\H[2]$                 \\
	$\Cliff_{\ev}(6,1)$  & $\iso$ & $\H[4]$    \\
	$\Cliff_{\ev}(10,1)$ & $\iso$ & $\R[32]$ \\
	\hline
\end{tabular}
\renewcommand{\arraystretch}{1}
\end{center}
These algebras have irreducible representations on $\R^4 \iso \C^2$, 
$\C^4 \iso \H^2$, $\H^4$ and $\O^4 \iso \R^{32}$, respectively.  

The details can be described in a uniform way for all four cases.  We
take as our space of `vectors' the following $(n+3)$-dimensional
subspace of $\K[4]$:
\[ \V = 
\left\{ 
\left( 
\begin{array}{cc} 
	a & \tilde{A} \\
	A & -a
\end{array}
\right)
: a \in \R, \quad A \in \h_2(\K)
\right\}
\]
In the last section, we defined vectors in $n+2$ dimensions to be $V =
\h_2(\K)$. That space has an obvious embedding into $\V$, given by
\[ \begin{array}{rcl}
	V & \inclusion & \V \\
	A & \mapsto    & \left( \begin{matrix} 0 & \tilde{A} \\ A & 0 \end{matrix} \right) 
	\end{array} \]
The Minkowski metric 
\[ h \maps \V \tensor \V \to \R \]
is given by extending the Minkowski metric $g$ on $V$:
\[ h \left( \left( \begin{smallmatrix} a & \tilde{A} \\ A & -a \end{smallmatrix} \right), \left( \begin{smallmatrix} a & \tilde{A} \\ A & -a \end{smallmatrix} \right) \right) = g(A,A) + a^2 \]
From our formulas for $g$, we can derive formulas for $h$:

\begin{prop} \label{prop:metric2}
	For any vectors $\A,\B \in \V \subseteq \K[4]$, we have
	\[ \A^2 = h(\A, \A)1 \]
	and
	\[ \fourth \Retr(\A \B) = h(\A, \B). \]
\end{prop}
\begin{proof}
	For $\A = \left( \begin{smallmatrix} a & \tilde{A} \\ A & -a
	\end{smallmatrix} \right)$, it is easy to check:
	\[ \A^2 = \left( \begin{matrix} a^2 + \tilde{A}A & 0 \\ 0 & A\tilde{A} + a^2 \end{matrix} \right) \]
	By Proposition \ref{prop:metric}, we have $A \tilde{A} = \tilde{A} A =
	g(A,A)1$, and substituting this in establishes the first formula. The
	second formula follows from polarizing and taking the real trace of
	both sides.
\end{proof}

Define a space of `spinors' by $\S = S_+ \oplus S_- = \K^4$. To distinguish 
elements of $\V$ from elements of $\h_2(\K)$, we will denote them with
caligraphic letters like $\A$, $\B$, \dots.  Similarly, to distinguish
elements of $\S$ from $S_\pm$, we will denote them with capital Greek letters
like $\Psi$, $\Phi$, \dots.

Elements of $\V$ act on $\S$ by left multiplication:
\[ \begin{array}{rcl}
	\V \tensor \S & \to     & \S \\
	\A \tensor \Psi        & \mapsto & \A \Psi
\end{array} \]
We can dualize this to get a map:
\[ \begin{array}{cccl}
	\Gamma \maps & \V & \to     & \End(\S) \\
	             & \A & \mapsto & L_\A 
\end{array} \]
This induces the Clifford action of $\Cliff(\V)$ on $\S$.
Note that this $\Gamma$ is the same as the map in Proposition
\ref{prop:n2cliff} when we restrict to $V \subseteq \V$.

\begin{prop}
The vectors $\V \subseteq \K[4]$ act on the spinors $\S = \K^4$ via the map
\[ \Gamma \maps \V \to \End(\S) \]
given by 
\[   \Gamma(\A)\Psi = \A \Psi \]
Furthermore, $\Gamma(\A)$ satisfies the Clifford algebra relation:
\[ \Gamma(\A)^2 = h(\A,\A) 1 \]
and so extends to a homomorphism $\Gamma \maps \Cliff(\V) \to \End(\S)$, i.e.\ a
representation of the Clifford algebra $\Cliff(\V)$ on $\S$.
\end{prop}
\begin{proof}
Here, we must be mindful of nonassociativity. For $\Psi = (\psi, \phi) \in
\S$ and $\A = \left( \begin{smallmatrix} a & \tilde{A} \\ A &
-a \end{smallmatrix} \right) \in \V$, we have:
\[ \Gamma(\A)^2 \Psi = \A ( \A \Psi) \]
which works out to be:
\[ \Gamma(\A)^2 \Psi = \left( \begin{array}{c} a^2 \psi + \tilde{A}(A\psi) \\ A(\tilde{A}\phi) + a^2 \phi \end{array} \right). \]
A quick calculation shows that the expressions $\tilde{A}(A\psi)$ and
$A(\tilde{A} \phi)$ involve at most two nonreal elements of $\K$, so everything
associates and we can write: 
\[ \Gamma(\A)^2 \Psi = \A^2 \Psi \]
By Proposition \ref{prop:metric2}, we are done.
\end{proof}

This tells us how $\S$ is a module of $\Cliff(\V)$, and thus
a representation of $\Spin(\V)$, the subgroup of $\Cliff(\V)$
generated by products of pairs of unit vectors. 

In the last section, we saw how to construct a $\Spin(V)$-invariant pairing
\[ \langle -,- \rangle \maps S_+ \tensor S_- \to \R. \]
We can use this to build up to a $\Spin(\V)$-invariant pairing on
$\S$:
\[ \langle (\psi, \phi), (\chi, \theta) \rangle = \langle \chi, \phi \rangle - \langle \psi, \theta \rangle \]
To see this, let
\[ \Gamma^0 = \left( \begin{matrix} 0 & -1 \\ 1 & 0 \end{matrix} \right) \]
Then, because $\langle \psi, \phi \rangle = \Re(\psi^\dagger \phi)$, it is easy
to check that:
\[ \langle \chi, \phi \rangle - \langle \psi, \theta \rangle = \Re \left( \left( \begin{matrix} \psi \\ \phi \end{matrix} \right)^\dagger \Gamma^0 \left( \begin{matrix} \chi \\ \theta \end{matrix} \right) \right). \]
We can show this last expression is invariant by explicit calculation.
\begin{prop}
	Define the nondegenerate skew-symmetric bilinear form
	\[ \langle -,- \rangle \maps \S \tensor \S \to \R \]
	by
	\[ \langle \Psi, \Phi \rangle = \Re(\Psi^\dagger \Gamma^0 \Phi). \]
	This form is invariant under $\Spin(\V)$.
\end{prop}
\begin{proof}
	It is easy to see that, for any spinors $\Psi, \Phi \in \S$
	and vectors $\A \in \V$, we have
	\[ \langle \A \Psi, \A \Phi \rangle = \Re \left( ( \Psi^\dagger \A^\dagger ) \Gamma^0 ( \A \Phi ) \right) = \Re \left( \Psi^\dagger ( \A^\dagger \Gamma^0 ( \A \Phi ) ) \right)\]
	where in the last step we have used Proposition \ref{prop:realtrace}. Now, given that
	\[ \A = \left( \begin{matrix} a & \tilde{A} \\ A & -a \end{matrix} \right) \]
	a quick calculation shows:
	\[ \A^\dagger \Gamma^0 = -\Gamma^0 \A. \]
	So, this last expression becomes:
	\[ -\Re \left( \Psi^\dagger ( \Gamma^0 \A ( \A \Phi ) ) \right) = -\Re \left( \Psi^\dagger ( \Gamma^0 \Gamma(\A)^2 \Phi ) ) \right) = -|\A|^2 \Re \left( \Psi^\dagger \Gamma^0 \Phi \right) \]
	where in the last step we have used the Clifford relation. Summing up, we have shown: 
	\[ \langle \A \Psi, \A \Phi \rangle = -|\A|^2 \langle \Psi, \Phi \rangle \]
	In particular, when $\A$ is a unit vector, acting by
	$\A$ changes the sign at most. Thus, $\langle -,- \rangle$ is
	invariant under the group generated by products of pairs of unit
	vectors, which is $\Spin(\V)$. It is easy to see that it is
	nondegenerate, and it is skew-symmetric because $\Gamma^0$ is.
\end{proof}

With the form $\langle -,- \rangle$ in hand, there is a manifestly equivariant
way to turn a pair of spinors into a vector.  Given $\Psi, \Phi \in
\S$, there is a unique vector $\Psi \cdot \Phi$ whose inner product
with any vector $\A$ is given by
\[ h(\Psi \cdot \Phi, \A) = \langle \Psi, \Gamma(\A) \Phi \rangle .\]

It will be useful to have an explicit formula for this operation:

\begin{prop} 
\label{prop:dot}
Given $\Psi = (\psi_1, \psi_2)$ and $\Phi = (\phi_1, \phi_2)$ in
$\S = S_+ \oplus S_-$, we have:
\[ \Psi \cdot \Phi = \left( 
\begin{array}{cc}
	\langle \psi_1, \phi_2 \rangle + \langle \phi_1, \psi_2 \rangle & -\widetilde{\psi_1 \cdot \phi_1} + \widetilde{\psi_2 \cdot \phi_2} \\
	-\psi_1 \cdot \phi_1 + \psi_2 \cdot \phi_2                      & -\langle \psi_1, \phi_2 \rangle - \langle \phi_1, \psi_2 \rangle \\
\end{array} \right)
\]
\end{prop}

\begin{proof}
	Decompose $\V$ into orthogonal subspaces:
	\[ \V = \left\{ \left( \begin{matrix} 0 & \tilde{A} \\ A & 0 \end{matrix} \right) : A \in V \right\} \oplus \left\{ \left( \begin{matrix} a & 0 \\ 0 & -a \end{matrix} \right) : a \in \R \right\} \]
	The first of these is just a copy of $V$, an $(n+2)$-dimensional
	Minkowski spacetime. The second is the single extra spatial dimension
	in our $(n+3)$-dimensional Minkowski spacetime, $\V$.

	Now, use the definition of $\Psi \cdot \Phi$, but restricted to $V$. It
	is easy to see that, for any vector $A \in V$, we have:
	\[ h(\Psi \cdot \Phi, A) = -\langle \psi_1, \gamma(A) \phi_1 \rangle + \langle \tilde{\gamma}(A) \psi_2, \phi_2 \rangle \]
	Letting $B$ be the component of $\Psi \cdot \Phi$ which lies in $V$, this becomes:
	\[ g(B, A) =  -\langle \psi_1, \gamma(A) \phi_1 \rangle + \langle \tilde{\gamma}(A) \psi_2, \phi_2 \rangle. \]
	Note that we have switched to the metric $g$ on $V$, to which $h$
	restricts. By definition, this is the same as:
	\[ g(B, A) =  g(-\psi_1 \cdot \phi_1 + \psi_2 \cdot \phi_2, A). \]
	Since this holds for all $A$, we must have $B = -\psi_1 \cdot \phi_1 + \psi_2 \cdot \phi_2$.

	It remains to find the component of $\Psi \cdot \Phi$ orthogonal to
	$B$. Since $\left\{ \left( \begin{smallmatrix} a & 0 \\ 0 & -a
	\end{smallmatrix} \right) : a \in \R \right\}$ is 1-dimensional, this
	is merely a number. Specifically, it is the constant of proportionality
	in the expression:
	\[ h \left(\Psi \cdot \Phi, \left( \begin{smallmatrix} a & 0 \\ 0 & -a \end{smallmatrix} \right) \right) = a ( \langle \psi_1, \phi_2 \rangle + \langle \phi_1, \psi_2 \rangle ) \]
	Thus, this component is $\langle \psi_1, \phi_2 \rangle + \langle \phi_1, \psi_2 \rangle$. Putting everything together, we get 
	\[ \Psi \cdot \Phi = \left(
          \begin{array}{cc}
            \langle \psi_1, \phi_2 \rangle + \langle \phi_1, \psi_2 \rangle & -\widetilde{\psi_1 \cdot \phi_1} + \widetilde{\psi_2 \cdot \phi_2} \\
            -\psi_1 \cdot \phi_1 + \psi_2 \cdot \phi_2                      & -\langle \psi_1, \phi_2 \rangle - \langle \phi_1, \psi_2 \rangle \\
          \end{array} \right)
	\]
\end{proof}

\section{The 4-$\Psi$'s Rule} \label{sec:4psi}

Spinors in dimension 4, 5, 7 and 11 satisfy an identity, written in
conventional notation as follows:
\[ \Psibar \Gamma_{ab} \Psi \Psibar \Gamma^b \Psi = 0 \]
This identity shows up in two prominent places in the physics literature.
First, it is required for the existence of 2-brane theories in these dimensions
\cite{AchucarroEvans, Duff}. This is because it allows the construction of a
Wess--Zumino--Witten term for these theories, which give these theories
Siegel symmetry.  

Yet it is known that 2-branes in 11 dimensions are intimately connected to
supergravity. Indeed, the Siegel symmetry imposed by the WZW term constrains
the 2-brane background to be that of 11-dimensional supergravity \cite{Tanii}.
So it should come as no surprise that this spinor identity also plays a
crucial role in supergravity, most visibly in the work of D'Auria and 
Fr\'e \cite{DAuriaFre} and subsequent work by Sati, Schreiber and Stasheff 
\cite{SSS}.   For previous work on its relation to division algebras, see
Foot and Joshi \cite{FJ}. 

This identity is equivalent to the `4-$\Psi$'s rule':
\[ \Psi \cdot ((\Psi \cdot \Psi) \Psi) = 0 . \]
To see this, note that we can turn a pair of spinors $\Psi$ and $\Phi$ into a
2-form, $\Psi * \Phi$.  This comes from the fact that we can embed bivectors inside
the Clifford algebra $\Cliff(\V)$ via the map
\[ \A \wedge \B \mapsto \A\B - \B\A \in \Cliff(\V). \]
These can then act on spinors using the Clifford action. Thus, define:
\begin{equation}
\label{eq:star}
 (\Psi * \Phi) (\A, \B) = \langle \Psi, (\A \wedge \B) \Phi \rangle. 
\end{equation}
But when $\Psi = \Phi$, we can simplify this using the Clifford relation:
\begin{eqnarray*}
	(\Psi * \Psi) (\A, \B) & = & \langle \Psi, (\A \B - \B \A) \Psi \rangle \\
	                       & = & \langle \Psi, 2\A \B \Psi \rangle - 2 \langle \Psi, \Psi \rangle h(\A, \B) \\
			       & = & 2 \langle \Psi, \A \B \Psi \rangle
\end{eqnarray*}
where we have used the skew-symmetry of the form. The index-ridden identity
above merely says that inserting the vector $\Psi \cdot \Psi$ into one slot
of the 2-form $\Psi * \Psi$ is zero, no matter what goes into the other slot:
\[ (\Psi * \Psi) (\A, \Psi \cdot \Psi) = 2 \langle \Psi, \A(\Psi \cdot \Psi) \Psi \rangle  = 0\]
for all $\A$. By the definition of the $\cdot$ operation, this is the same as
\[ 2 h( \Psi \cdot ( (\Psi \cdot \Psi) \Psi), \A) = 0 \]
for all $\A$. Thus, the index-ridden identity is equivalent to:
\[ \Psi \cdot ((\Psi \cdot \Psi) \Psi) = 0 \]
as required.

Now, let us prove this:
\begin{thm}
Suppose $\Psi \in \S$.  Then $\Psi \cdot ((\Psi \cdot \Psi) \Psi) = 0$. 
\end{thm}

\begin{proof}
	Let $\Psi = (\psi, \phi)$.  By Proposition \ref{prop:dot},
	\[ \Psi \cdot \Psi = \left( \begin{array}{cc} 2 \langle \psi, \phi \rangle       & -\widetilde{\psi \cdot \psi} + \widetilde{\phi \cdot \phi} \\
	                                              -\psi \cdot \psi + \phi \cdot \phi & -2 \langle \psi, \phi \rangle  \\
					      \end{array} \right) \]
and thus
	\[ (\Psi \cdot \Psi) \Psi = \left( \begin{array}{c} 2 \langle \psi, \phi \rangle \psi - (\widetilde{\psi \cdot \psi}) \phi + (\widetilde{\phi \cdot \phi}) \phi \\
	                                                    -(\psi \cdot \psi)\psi + (\phi \cdot \phi)\psi - 2\langle \psi, \phi \rangle \phi  \\
					      \end{array} \right). \]
	Both $(\psi \cdot \psi)\psi = 0$ and 
$(\widetilde{\phi \cdot \phi}) \phi = 0$ by the 3-$\psi$'s rule,
Theorem \ref{thm:fundamental_identity}.
	So:
	\[ (\Psi \cdot \Psi) \Psi = \left( \begin{array}{c} 2 \langle \psi, \phi \rangle \psi - (\widetilde{\psi \cdot \psi}) \phi \\
	                                                    (\phi \cdot \phi)\psi - 2\langle \psi, \phi \rangle \phi  \\
					      \end{array} \right). \]
	The resulting matrix for $\Psi \cdot ((\Psi \cdot \Psi) \Psi)$ is large
	and unwieldy, so we shall avoid writing it out. Fortunately, all
	we really need is the $(1,1)$ entry. Recall, this is the component of the
	vector $\Psi \cdot ((\Psi \cdot \Psi) \Psi)$ that is orthogonal to the
	subspace $V \subset \V$. Call this component $a$.  A calculation shows:
	\begin{eqnarray*}
       a & = & \langle \psi, (\phi \cdot \phi)\psi \rangle - \langle (\widetilde{\psi \cdot \psi}) \phi, \phi \rangle \\
		  & = & \Retr (\psi^\dagger (2 \phi \phi^\dagger) \psi)  - \Retr (\phi^\dagger (2 \psi \psi^\dagger) \phi ) \\
		  & = & 0
	\end{eqnarray*}
where the two terms cancel by the cyclic property of the real trace, 
Proposition \ref{prop:realtrace}.   Thus, this
component of the vector $\Psi \cdot ((\Psi \cdot \Psi) \Psi)$ vanishes.
But since the map $\Psi \mapsto \Psi \cdot ((\Psi \cdot \Psi) \Psi)$ is 
equivariant with respect to the action of $\Spin(\V)$, and $\V$ is an 
irreducible representation of this group, it follows that all components of 
this vector must vanish.
\end{proof}

\section{Cohomology of Lie Superalgebras} \label{sec:cohomology}

In this section we explain more of the meaning of the 3-$\psi$'s rule
and 4-$\Psi$'s rules: they are \textit{cocycle conditions}.   In any
dimension, a symmetric bilinear intertwining operator that eats two 
spinors and spits out a vector gives rise to a `super-Minkowski
spacetime'.  The infinitesimal translation symmetries of this object
form a Lie superalgebra called the `supertranslation algebra'
\cite{Deligne}.  The cohomology of this Lie superalgebra
is interesting and apparently rather subtle.  We shall see that its
3rd cohomology is nontrivial in dimensions 3, 4, 6 and 10, thanks to
the 3-$\psi$'s rule.  Similarly, its 4th cohomology is nontrivial
in dimensions 4, 5, 7 and 11, thanks to the 4-$\Psi$'s rule. 

We begin by recalling the supertranslation algebra.  Take $V$ to
be the space of vectors in Minkowski spacetime in any dimension,
and take $S$ to be any spinor representation in this dimension.
Suppose that there is a symmetric equivariant bilinear map:
\[ \cdot \maps S \tensor S \to V. \]
Form a super vector space $\T$ with 
\[ \T_0 = V , \qquad \T_1 = S .\]
We make $\T$ into a Lie superalgebra, the \textbf{supertranslation
algebra}, by giving it a suitable bracket operation.  This bracket
will be zero except when we bracket 
a spinor with a spinor, in which case it is simply the operation 
\[ \cdot \maps S \tensor S \to V. \]
Since this is symmetric and spinors are odd, the bracket operation is
super-skew-symmetric overall. Furthermore, the Jacobi identity holds trivially,
thanks to the near triviality of the bracket.  Thus $\T$ is 
indeed, a Lie superalgebra.

Despite the fact that $\T$ is nearly trivial, its cohomology is not.
To see this, we must first recall how to generalize
Chevalley--Eilenberg cohomology \cite{AzcarragaIzquierdo, 
ChevalleyEilenberg} from Lie algebras to Lie
superalgebras.  Suppose $\g$ is a Lie superalgebra and $R$ is a
representation of $\g$.  That is, $R$ is a supervector space equipped
with a Lie superalgebra homomorphism $\rho \maps \g \to \gl(X)$.  We
now define the cohomology groups of $\g$ with values in $R$.

First, of course, we need a cochain complex. We define the
\define{$n$-cochains} $C^n(\g, R)$ to be the vector space of
super-skew-symmetric $n$-linear maps:
\[ \Lambda^n \g \to R. \]
In fact, the $n$-cochains $C^n(\g, R)$ are a super vector space, in which
parity-preserving elements are even, while parity-reversing elements are odd.

Next, we define the coboundary operator $d \maps C^n(\g, R) \to
C^{n+1}(\g, R)$. Let $\omega$ be a homogeneous $n$-cochain and let $X_1, \dots,
X_{n+1}$ be homogeneous elements of $\g$. Now define:
\begin{eqnarray*}
& & d\omega(X_1, \dots, X_{n+1}) = \\ 
& & \sum^{n+1}_{i=1} (-1)^{i+1} (-1)^{|X_i||\omega|} \epsilon^{i-1}_1(i) \rho(X_i) \omega(X_1, \dots, \hat{X}_i, \dots, X_{n+1}) \\
& & + \sum_{i < j} (-1)^{i+j} (-1)^{|X_i||X_j|} \epsilon^{i-1}_1(i) \epsilon^{j-1}_1(j) \omega([X_i, X_j], X_1, \dots, \hat{X}_i, \dots, \hat{X}_j, \dots X_{n+1})
\end{eqnarray*}
Here, $\epsilon^j_i(k)$ is shorthand for the sign one obtains by moving $X_k$
through $X_i, X_{i+1}, \dots, X_j$. In other words,
\[ \epsilon^j_i(k) = (-1)^{|X_k|(|X_i| + |X_{i+1}| + \dots + |X_j|)}. \]

Following the usual argument for Lie algebras, one can check that:

\begin{prop}
	The Lie superalgebra coboundary operator $d$ satisfies $d^2 = 0$.
\end{prop}

\noindent
We thus say a $R$-valued $n$-cochain $\omega$ on $\g$ is an
\define{$n$-cocycle} or \define{closed} when $d \omega = 0$, and an
\define{$n$-coboundary} or \define{exact} if there exists an $(n-1)$-cochain
$\theta$ such that $\omega = d \theta.$   Every $n$-coboundary is an
$n$-cocycle, and we say an $n$-cocycle is \define{trivial} if it is a
coboundary.  We denote the super vector spaces of $n$-cocycles and
$n$-coboundaries by $Z^n(\g,V)$ and $B^n(\g,V)$ respectively.  The $n$th {\bf
Lie superalgebra cohomology of $\g$ with coefficients in $R$}, denoted
$H^n(\g,R)$ is defined by 
\[ H^n(\g,R) = Z^n(\g,R)/B^n(\g,R). \]
This super vector space is nonzero if and only if there is a nontrivial
$n$-cocycle. In what follows, we shall be especially concerned with the even
part of this super vector space, which is nonzero if and only if there is a
nontrivial even $n$-cocycle. Our motivation for looking for even cocycles is
simple: these parity-preserving maps can regarded as morphisms in the category
of super vector spaces, which is crucial for the construction in Theorem
\ref{trivd} and everything following it.

Now consider Minkowski spacetimes of dimensions 3, 4, 6, and 10.  Here
Minkowski spacetime can be written as $V = \h_2(\K)$, and we can take our
spinors to be $S_+ = \K^2$.  Since from Section \ref{sec:n+2} we know there is
a symmetric bilinear intertwiner $\cdot \maps S_+ \otimes S_+ \to V$, we obtain
the supertranslation algebra $\T = V \oplus S_+$.  We can decompose the space
of $n$-cochains with coefficients in the trivial 1-dimensional representation
of $\T$ into summands by counting how many of the arguments are vectors and how
many are spinors:
\[     C^n(\T,\R) \iso  
 \bigoplus_{p + q = n}   (\Lambda^p(V) \otimes \Sym^q(S_+))^* . \]
We call an element of $(\Lambda^p(V) \otimes \Sym^q(S_+))^*$ a 
{\bf {\boldmath$(p,q)$-form}}.  Since the bracket of two spinors is
a vector, and all other brackets are zero, $d$ of a $(p,q)$-form is
a $(p-1,q+2)$-form.  

Using the 3-$\psi$'s rule we can show:

\begin{thm}
\label{thm:3-cocycle}
      In dimensions 3, 4, 6 and 10, the supertranslation algebra $\T$
	has a nontrivial even 3-cocycle taking values in the trivial 
        representation $\R$, namely the unique $(1,2)$-form with 
	\[ \alpha(\psi, \phi, A) = g(\psi \cdot \phi, A) \]
	for spinors $\psi, \phi \in S_+$ and vectors $A \in V$.
    \end{thm}

\begin{proof}
	First, note that $\alpha$ has the right symmetry to be a linear map on
	$\Lambda^3(V \oplus S_+)$. Second, note that $\alpha$ is a
	$(1,2)$-form, eating one vector and two spinors. Thus $d\alpha$ is a
	$(0,4)$-form.

	Because spinors are odd, $d\alpha$ is a symmetric function of four
	spinors. By the definition of $d$, $d\alpha(\psi, \phi, \chi, \theta)$
	is the totally symmetric part of $\alpha(\psi \cdot \phi, \chi, \theta)
	= \alpha(\chi, \theta, \psi \cdot \phi) = g(\chi \cdot \theta, \psi
	\cdot \phi)$.   But any symmetric 4-linear form can be obtained from
	polarizing a quartic form. In this, we polarize $g(\psi \cdot \psi,
	\psi \cdot \psi)$ to get $d\alpha$. Thus:
	\[ d\alpha(\psi, \psi, \psi, \psi) = g(\psi \cdot \psi, \psi \cdot \psi) = \langle \psi, (\psi \cdot \psi) \psi \rangle \]
	where we have used the definition of the dot operation to obtain the last
	expression, which vanishes due to the 3-$\psi$ rule. Thus 
	$\alpha$ is closed.

	It remains to show $\alpha$ is not exact. So suppose it is exact, and that
	\[ \alpha = d\omega .\]
	By our remarks above we may assume $\omega$ is a $(2,0)$-form: that is, an 
	antisymmetric bilinear function of two vectors. By the definition of 
      $d$, this last equation says:
	\[ g(\psi \cdot \phi, A) = -\omega(\psi \cdot \phi, A). \]
	But since $S_+ \tensor S_+ \to V$ is onto, this implies
	\[ g = -\omega, \]
	a contradiction, since $g$ is symmetric while $\omega$ is
	antisymmetric.	
\end{proof}

Next consider Minkowski spacetimes of dimensions 4, 5, 7 and 11.  In this case
Minkowski spacetime can be written as a subspace $\V$ of the $4 \times 4$
matrices valued in $\K$, and we can take our spinors to be $\S = \K^4$.  Since
from Section \ref{sec:n+3} we know there is a symmetric bilinear intertwiner
$\cdot \maps \S \otimes \S \to \V$, we obtain a supertranslation algebra $\T =
\V \oplus \S$.  As before, we can uniquely 
decompose any $n$-cochain in $C^n(\T, \R)$ into a sum of 
$(p,q)$-forms, where now a {\bf {\boldmath$(p,q)$-form}} is an 
an element of $(\Lambda^p(\V) \otimes \Sym^q(\S))^*$. 
As before, $d$ of a $(p,q)$-form is a $(p-1,q+2)$-form.  
And using the 4-$\Psi$'s rule, we can show:

\begin{thm}
\label{thm:4-cocycle}
	In dimensions 4, 5, 7 and 11, the supertranslation algebra $\T$
	has a nontrivial even 4-cocycle, namely the unique $(2,2)$-form with
	\[ \beta(\Psi, \Phi, \A, \B) = \langle \Psi, (\A\B - \B\A) \Phi \rangle \]
	for spinors $\Psi, \Phi \in \S$ and vectors $\A, \B \in \V$.
        Here the commutator $\A\B - \B\A$ is taken in the Clifford algebra
        of $\V$.
\end{thm} 

\begin{proof}
	First, to see that $\beta$ has the right symmetry to be a map on
	$\Lambda^4(\V \oplus \S)$, we note that it is antisymmetric on vectors,
	and that because
	\[ \Gamma^0 \A = -\A^\dagger \Gamma^0, \]
	we have:
	\[ \Gamma^0 \A \B = \A^\dagger \B^\dagger \Gamma^0. \]
	Thus:
	\[ \langle \Psi, \A \B \Phi \rangle = \langle \B \A \Psi, \Phi \rangle = -\langle \Phi, \B \A \Psi \rangle, \]
	so we have:
	\[ \langle \Psi, (\A\B - \B\A) \Phi \rangle = \langle \Phi, (\A\B - \B\A) \Psi \rangle. \]
	Thus, $\beta$ is symmetric on spinors.

	Next note that $d\beta$ is a $(1,4)$-form, symmetric on its four spinor
	inputs. It is thus proportional to the polarization of
	\[ \beta(\Psi, \Psi, (\Psi \cdot \Psi), \A) = 
         \Psi * \Psi( \Psi \cdot \Psi, \A) \]
	We encountered this object in Section~\ref{sec:4psi}, where we showed
	that it is proportional to 
	\[ h(\Psi \cdot [ (\Psi \cdot \Psi) \Psi], \A). \]
	Moreover, this last expression vanishes by the 4-$\Psi$'s rule.  So,
	$\beta$ is closed.

	Furthermore, $\beta$ is not exact. To see this, consider the unit
	vector $\left( \begin{smallmatrix} 1 & 0 \\ 0 & -1 \end{smallmatrix}
	\right)$ orthogonal to $V \subseteq \V$. Taking the interior product of
	$\beta$ with this vector, a quick calculation shows:
	\[ \beta(\Psi, \Phi, \left( \begin{smallmatrix} 1 & 0 \\ 0 & -1 \end{smallmatrix} \right), \A) = 2\langle \psi_1, \gamma(A) \phi_1 \rangle + 2\langle \tilde{\gamma}(A) \psi_2, \phi_2 \rangle, \]
	where we have decomposed $\Psi = (\psi_1,\psi_2)$ and $\Phi = (\phi_1,
	\phi_2)$ into their components in $\S = S_+ \oplus S_-$, and $A$ is the
	component of $\A$ in $V$. Restricting to the subalgebra $V \oplus S_+
	\subseteq \V \oplus \S$, we see this is just $\alpha$, up to a factor.

	So, it suffices to check that interior product with $X = \left(
	\begin{smallmatrix} 1 & 0 \\ 0 & -1 \end{smallmatrix} \right)$
	preserves exactness. For then, if $\beta$ were exact, it would
	contradict that fact that $\alpha$ is not. Indeed, let $\omega$ be an
	$n$-cochain on $\T$, and let $X_1, \dots, X_n \in \T$. Then, by our
	formula for the coboundary operator, we have:
	\begin{eqnarray*}
	& & d\omega(X, X_1, \dots, X_n)  = \\
	& & \sum_{i < j} -(-1)^{i+j} (-1)^{|X_i||X_j|} \epsilon^{i-1}_1(i) \epsilon^{j-1}_1(j) \omega(X, [X_i, X_j], X_1, \dots, \hat{X}_i, \dots, \hat{X}_j, \dots X_{n}) \\
	& & + \sum_{i=1}^n (-1)^{1+i} \epsilon_1^{i-1}(i) \omega([X,X_i], X_1, \dots, \hat{X}_i, \dots, X_n),
	\end{eqnarray*}
	where, taking care with signs, we have collected terms involving
	bracketing with $X$ into the second summation. But $X$ is a vector, so
	all brackets with it vanish, and the second summation is zero.

	If we write $i_X \omega$ for the operation of taking the interior
	product of $\omega$ with $X$, we have just shown:
	\[ i_X d \omega = -d i_X \omega \]
	for any $\omega$. In particular, if $\omega = d \theta$ then $i_X
	\omega = d( -i_X \theta )$, and so interior product with $X$ preserves
	exactness, as claimed.
\end{proof}

\section{$L_\infty$-superalgebras} \label{sec:Linfinity}

In the last section, we saw that the 3-$\psi$'s and 4-$\Psi$'s rules are
cocycle conditions for the cocycles $\alpha$ and $\beta$. This sheds some light
on the meaning of these rules, but it prompts an obvious followup question:
what are these cocycles good for?  

There is a very general answer to this question: a cocycle on a Lie
superalgebra lets us extend it to an `$L_\infty$-superalgebra'.  As we touched
on in the Introduction, this is a chain complex equipped with structure like
that of a Lie superalgebra, but where all the laws hold only `up to chain
homotopy'.  We give the precise definition below.

It is well known that that the 2nd cohomology of a Lie algebra $\g$
with coefficients in some representation $R$ classifies `central
extensions' of $\g$ by $R$ \cite{AzcarragaIzquierdo, ChevalleyEilenberg}.
These are short exact sequences of Lie algebras:
\[         0 \to R \to \tilde{\g} \to \g \to 0  \]
where $\tilde{g}$ is arbitrary and $R$ is treated as an abelian Lie algebra
whose image lies in the center of $\tilde{g}$.
The same sort of result is true for Lie superalgebras.  But this is just a
special case of an even more general fact.

Suppose $\g$ is a Lie superalgebra with a representation on a
supervector space $R$.  Then we shall prove that an even $R$-valued
$(n+2)$-cocycle $\omega$ on $\g$ lets us construct an 
$L_\infty$-superalgebra, say $\tilde{\g}$, of the following form:
\[  
\g \stackrel{d}{\longleftarrow} 0 
\stackrel{d}{\longleftarrow} \dots 
\stackrel{d}{\longleftarrow} 0 
\stackrel{d}{\longleftarrow} R .
\]
where only the 0th and and $n$th grades are nonzero.   Moreover,
$\tilde{\g}$ is an \define{extension} of $\g$: there
is a short exact sequence of $L_\infty$-superalgebras 
\[  0 \to \b^n R \to \tilde{\g} \to \g \to 0 . \]
Here $\b^n R$ is the abelian $L_\infty$-superalgebra
with $R$ as its $n$th grade and all the rest zero:
\[             
0 
\stackrel{d}{\longleftarrow} 0 
\stackrel{d}{\longleftarrow} \cdots
\stackrel{d}{\longleftarrow} 0 
\stackrel{d}{\longleftarrow} \R   \]
Note that when $n = 0$ and our vector spaces are all purely even, 
we are back to the familiar construction of Lie algebra extensions 
from 2-cocycles.

Technically, we should be more general than this in defining extensions.
Maps between $L_\infty$-algebras admit homotopies among themselves, and
this allows us to introduce a weakened notion of ``short exact sequence'': 
a fibration sequence in the $(\infty,1)$-category of $L_\infty$-algebras.
In general, these fibration sequences give the right concept of 
extension for $L_\infty$-algebras.  However, for the very special 
extensions we consider here, ordinary short exact sequences are all we 
need. 

It is useful to have a special name for $L_\infty$-superalgebras
whose nonzero terms are all of degree $< n$: we call them
\define{{\boldmath Lie $n$-superalgebras}}.  In this
language, the 3-cocycle $\alpha$ defined in Theorem \ref{thm:3-cocycle}
gives rise to a Lie 2-superalgebra
\[              \T \stackrel{d}{\longleftarrow} \R   \]
extending the supertranslation algebra $\T$ in dimensions $3,4,6,$ and $10$.
Similarly, the 4-cocycle $\beta$ defined in Theorem \ref{thm:4-cocycle}
gives a Lie 3-superalgebra 
\[              \T \stackrel{d}{\longleftarrow} 0 
                   \stackrel{d}{\longleftarrow} \R   \]
extending the supertranslation algebra in dimensions $4,5,7$ and $11$.

Of course, this raises yet another question: what are these Lie
$n$-superalgebras good for? The answer lies in physics: following a suggestion
of Urs Schreiber, we can apply his work with Sati and Stasheff \cite{SSS} and
consider connections on `$n$-bundles'.  These are, roughly speaking, bundles
where the fibers are smooth $n$-categories instead of smooth manifolds.  If we
do this for $n$-bundles---or more precisely, `super-$n$-bundles'---arising from
the Lie $n$-superalgebras constructed here, we find that the connections are
fields that show up naturally in supermembrane and supergravity theories.  

But now let us turn to the business at hand.  In what follows, we shall
use \define{super chain complexes}, which are chain complexes in the
category SuperVect of $\Z_2$-graded vector spaces:
\[  V_0 \stackrel{d}{\longleftarrow}
    V_1 \stackrel{d}{\longleftarrow}
    V_2 \stackrel{d}{\longleftarrow} \cdots \]
Thus each $V_p$ is $\Z_2$-graded and $d$ preserves this grading.

There are thus two gradings in play: the $\Z$-grading by
\define{degree}, and the $\Z_2$-grading on each vector space, which we
call the \define{parity}. We shall require a sign convention to
establish how these gradings interact. If we consider an object of odd
parity and odd degree, is it in fact even overall?  By convention, we
assume that it is. That is, whenever we interchange something of
parity $p$ and degree $q$ with something of parity $p'$ and degree
$q'$, we introduce the sign $(-1)^{(p+q)(p'+q')}$. We shall call the
sum $p+q$ of parity and degree the \define{overall grade}, or when it
will not cause confusion, simply the grade. We denote the overall
grade of $X$ by $|X|$.

We require a compressed notation for signs. If $x_{1}, \ldots, x_{n}$ are
graded, $\sigma \in S_{n}$ a permutation, we define the \define{Koszul sign}
$\epsilon (\sigma) = \epsilon(\sigma; x_{1}, \dots, x_{n})$ by 
\[ x_{1} \cdots x_{n} = \epsilon(\sigma; x_{1}, \ldots, x_{n}) \cdot x_{\sigma(1)} \cdots x_{\sigma(n)}, \]
the sign we would introduce in the free graded-commutative algebra generated by
$x_{1}, \ldots, x_{n}$. Thus, $\epsilon(\sigma)$ encodes all the sign changes
that arise from permuting graded elements. Now define:
\[ \chi(\sigma) = \chi(\sigma; x_{1}, \dots, x_{n}) := \textrm{sgn} (\sigma) \cdot \epsilon(\sigma; x_{1}, \dots, x_{n}). \]
Thus, $\chi(\sigma)$ is the sign we would introduce in the free
graded-anticommutative algebra generated by $x_1, \dots, x_n$.

Yet we shall only be concerned with particular permutations. If $n$ is a
natural number and $1 \leq j \leq n-1$ we say that $\sigma \in S_{n}$ is an
\define{$(j,n-j)$-unshuffle} if
\[ \sigma(1) \leq\sigma(2) \leq \cdots \leq \sigma(j) \hspace{.2in} \textrm{and} \hspace{.2in} \sigma(j+1) \leq \sigma(j+2) \leq \cdots \leq \sigma(n). \] 
Readers familiar with shuffles will recognize unshuffles as their inverses. A
\emph{shuffle} of two ordered sets (such as a deck of cards) is a permutation
of the ordered union preserving the order of each of the given subsets. An
\emph{unshuffle} reverses this process. We denote the collection of all
$(j,n-j)$ unshuffles by $S_{(j,n-j)}$.

The following definition of an $L_{\infty}$-algebra was formulated by
Schlessinger and Stasheff in 1985 \cite{SS}:

\begin{defn} \label{L-alg} An
{\bf $\mathbf{L_{\infty}}$-algebra} is a graded vector space $V$
equipped with a system $\{l_{k}| 1 \leq k < \infty\}$ of linear maps
$l_{k} \maps V^{\otimes k} \rightarrow V$ with $\deg(l_{k}) = k-2$
which are totally antisymmetric in the sense that
\begin{eqnarray}
   l_{k}(x_{\sigma(1)}, \dots,x_{\sigma(k)}) =
   \chi(\sigma)l_{k}(x_{1}, \dots, x_{n})
\label{antisymmetry}
\end{eqnarray}
for all $\sigma \in S_{n}$ and $x_{1}, \dots, x_{n} \in V,$ and,
moreover, the following generalized form of the Jacobi identity
holds for $0 \le n < \infty :$
\begin{eqnarray}
   \displaystyle{\sum_{i+j = n+1}
   \sum_{\sigma \in S_{(i,n-i)}}
   \chi(\sigma)(-1)^{i(j-1)} l_{j}
   (l_{i}(x_{\sigma(1)}, \dots, x_{\sigma(i)}), x_{\sigma(i+1)},
   \ldots, x_{\sigma(n)}) =0,}
\label{megajacobi}
\end{eqnarray}
where the summation is taken over all $(i,n-i)$-unshuffles with $i
\geq 1.$
\end{defn}

The following result shows how to construct $L_{\infty}$-superalgebras from Lie
superalgebra cocycles.  This is the `super' version of a result due to Crans
\cite{BaezCrans}. In this result, we require our cocycle to be even so we can
consider it as a morphism in the category of super vector spaces.

\begin{thm} \label{trivd}
There is a one-to-one correspondence between $L_{\infty}$-superalgebras
consisting of only two nonzero terms $V_{0}$ and $V_{n}$, with $d=0$, and
quadruples $(\g, V, \rho, l_{n+2})$ where $\g$ is a Lie superalgebra, $V$ is a
super vector space, $\rho$ is a representation of $\g$ on $V$, and $l_{n+2}$ is
an even $(n+2)$-cocycle on $\g$ with values in $V$.
\end{thm}

\begin{proof}

Given such an $L_{\infty}$-superalgebra we set $\g=
V_0$.  $V_0$ comes equipped with a bracket as part of the
$L_{\infty}$-structure, and since $d$ is trivial, this bracket satisfies the
Jacobi identity on the nose, making $\g$ into a Lie superalgebra. We define $V
= V_{n}$, and note that the bracket also gives a map $\rho \maps \g \tensor V
\to V$, defined by $\rho(x)f = [x,f]$ for $x \in \g, f \in V$. We have
\begin{eqnarray*}
  \rho ([x,y])f &=& [[x,y],f] \\
  		&=& (-1)^{|y||f|}[[x,f],y] + [x,[y,f]] \; \; \; \;
                    \textrm{by $(3)$ of Definition \ref{L-alg}} \\
                &=& (-1)^{|f||y|}[\rho(x)f, y] + [x, \rho(y) f] \\
                &=& -(-1)^{|x||y|}\rho(y)\rho(x)f + \rho(x)\rho(y) f \\
		&=& [\rho(x), \rho(y)]f
\end{eqnarray*}

\noindent
for all $x,y \in \g$ and $f \in V$, so that $\rho$ is indeed a representation.
Finally, the $L_{\infty}$ structure gives a map $l_{n+2} \maps
\Lambda^{n+2} \g \to V$ which is in fact an $(n+2)$-cocycle.  To see this, note
that
\[ 0 = \sum_{i+j = n+4} \sum_{\sigma} \chi(\sigma) (-1)^{i(j-1)} l_{j}(l_{i}(x_{\sigma(1)}, \ldots, x_{\sigma(i)}), x_{\sigma(i+1)}, \ldots, x_{\sigma(n+3)}) \]
where we sum over $(i, (n+3)-i)$-unshuffles $\sigma \in S_{n+3}$.  However, the
only choices for $i$ and $j$ that lead to nonzero $l_{i}$ and $l_{j}$ are
$i=n+2, j=2$ and $i=2, j=n+2$. Thus, the above becomes, with $\sigma$ a
$(n+2, 1)$-unshuffle and $\tau$ a $(2, n+1)$-unshuffle:
\begin{eqnarray*}
0 & = & \sum_{\sigma} \chi(\sigma) (-1)^{n+2} [l_{n+2}(x_{\sigma(1)}, \dots, x_{\sigma(n+2)}), x_{\sigma(n+3)}] \\
  &   & + \sum_{\tau} \chi(\tau) l_{n+2}([x_{\tau(1)}, x_{\tau(2)}], x_{\tau(3)}, \dots, x_{\tau(n+3)}) \\
  & = & \sum_{i=1}^{n+3} (-1)^{n+3-i}(-1)^{n+2} \epsilon^{n+2}_{i+1}(i) [l_{n+2}(x_1, \dots, \hat{x}_i, \dots, x_{n+3}), x_{i}] \\
  &   & + \sum_{1 \leq i < j \leq n+3} (-1)^{i+j+1} (-1)^{|x_i||x_j|} \epsilon^{i-1}_1(i) \epsilon^{j-1}_1(j) l_{n+2}([x_{i}, x_{j}], x_{1}, \dots,\hat{x}_i, \dots, \hat{x}_j, \dots, x_{n+3})
\end{eqnarray*}
On the second line, we have explicitly specified the unshuffles and unwrapped
the signs encoded by $\chi$. Since $l_{n+2}$ is a morphism in SuperVect, it
preserves parity, and thus the element $l_{n+2}(x_1, \dots, \hat{x}_i, \dots,
x_{n+2})$ has parity $|x_1| + \dots + |x_{i-1}| + |x_{i+1}| + \dots +
|x_{n+2}|$. So, we can reorder the bracket in the first term, at the cost of a
sign:
\begin{eqnarray*}
0 & = & \sum_{i=1}^{n+3} -(-1)^{i+1} \epsilon^{i-1}_1(i) [x_i, l_{n+2}(x_{1}, \dots, \hat{x}_i, \dots, x_{n+3})] \\
  &   & + \sum_{1 \leq i < j \leq n+3} -(-1)^{i+j} (-1)^{|x_i||x_j|} \epsilon^{i-1}_1(i) \epsilon^{j-1}_1(j) l_{n+2}([x_{i}, x_{j}], x_{1}, \dots, \hat{x}_i, \dots, \hat{x}_j, \dots, x_{n+3}) \\
  & = & -dl_{n+2} 
\end{eqnarray*}
Here, we have used the fact that $\epsilon_{i+1}^{n+2}(i) (-1)^{|x_i|(|x_1| +
\dots + |x_{i-1}| + |x_{i+1}| + \dots + |x_{n+2}|)} = \epsilon_1^{i-1}(i)$.
Thus, $l_{n+2}$ is indeed a cocycle.

Conversely, given a Lie superalgebra $\g$, a representation $\rho$ of $\g$ on a
vector space $V$, and an even $(n+2)$-cocycle $l_{n+2}$ on $\g$ with values in
$V$, we define our $L_{\infty}$-superalgebra $V$ by setting
$V_{0} = \g$, $V_{n} = V$, $V_i = \{0\}$ for $i \ne 0,n$,
and $d=0$.  It remains to define the system of linear maps
$l_{k}$, which we do as follows: Since $\g$ is a Lie
superalgebra, we have a bracket defined on $V_{0}$. We extend this
bracket to define the map $l_2$, denoted by $[\cdot, \cdot] \maps
V_{i} \otimes V_{j} \rightarrow V_{i+j}$ where $i,j=0,n,$ as
follows:
\[ [x,f] = \rho(x) f, \quad [f,y] = (-1)^{|y||f|} \rho(y)f, \quad [f,g] = 0 \]
for $x,y \in V_0$ and $f,g \in V_n$.
With this definition, the map $[\cdot, \cdot]$ satisfies condition $(1)$ of
Definition \ref{L-alg}. We define $l_{k}=0$ for $3 \leq k \leq n+1$ and $k>
n+2$, and take $l_{n+2}$ to be the given $(n+2)$ cocycle, which satisfies
conditions $(1)$ and $(2)$ of Definition \ref{L-alg} by the cocycle condition. 
\end{proof}

As already mentioned, we call an $L_\infty$-superalgebra 
whose nonzero terms are all of degree $< n$ a \define{\boldmath 
Lie $n$-superalgebra}.  Using the above theorem, we say a
Lie $n$-superalgebra is \define{exact} if only $V_0$ and $V_{n-1}$ are nonzero
and the $(n+1)$-cocycle $l_{n+1}$ is trivial.  The point of this definition is
that we may easily obtain an exact Lie $n$-superalgebra simply from a Lie
superalgebra and a representation, simply by taking the $(n+1)$-cocycle to be
zero.  Non-exact Lie $n$-superalgebras are more interesting.
\begin{cor}

\label{cor:1}
	In dimensions 3, 4, 6 and 10, there exists a non-exact Lie
	2-superalgebra corresponding to the cocycle $\alpha$ given in Theorem
\ref{thm:3-cocycle}.
\end{cor}

\begin{cor}
\label{cor:2}
	In dimensions 4, 5, 7 and 11, there exists a non-exact Lie
	3-superalgebra corresponding to the cocycle $\beta$ given in Theorem
	\ref{thm:4-cocycle}.
\end{cor}

\section{Superstring Lie 2-Algebras, 2-Brane Lie 3-algebras} 
\label{sec:algebras}

We have now met all the stars of our story. First we met the 3-cocycle $\alpha$:
\[ \begin{array}{cccc}
	\alpha \maps & \Lambda^3(\T)              & \to     &  \R \\
	             &  A \wedge \psi \wedge \phi & \mapsto & \langle \psi, A \phi \rangle \\
\end{array}
\]
and saw its cocycle condtion is the 3-$\psi$'s rule in
Theorem~\ref{thm:3-cocycle}. Next, we encountered the 4-cocycle
$\beta$:
\[ \begin{array}{cccc}
	\beta \maps & \Lambda^4(\T)                        & \to     & \R \\
	            & \A \wedge \B \wedge \Psi \wedge \Phi & \mapsto & \langle \Psi, (\A \wedge \B) \Phi \rangle \\
   \end{array}
\]
and saw that its cocycle condition is the 4-$\Psi$'s rule in
Theorem~\ref{thm:4-cocycle}. We explored the meaning of these cocycles in
Section~\ref{sec:Linfinity}: they allow us to create nontrivial extensions of
the supertranslations $\T$ to Lie 2- and Lie 3-superalgebras, which we
exhibited in Corollaries \ref{cor:1} and \ref{cor:2}. But that is not the end
of our story. We can go one step further with $\alpha$ and $\beta$, because
both of them are invariant under the action of the corresponding Lorentz
algebra: $\so(n+1,1)$ in the case of $\alpha$, and $\so(n+2,1)$ for $\beta$.
This is manifestly true, because $\alpha$ and $\beta$ are built from
equivariant maps.

As we shall see, this invariance implies that $\alpha$ and $\beta$ are
cocycles, not merely on the supertranslations, but on the full Poincar\'e
superalgebra---$\siso(n+1,1)$ in the case of $\alpha$:
\[ \siso(n+1,1) = \so(n+1,1) \ltimes \T \]
and $\siso(n+2,1)$ in the case of $\beta$:
\[ \siso(n+2,1) = \so(n+2,1) \ltimes \T \]
We can extend $\alpha$ and $\beta$ to these larger algebras in a trivial way:
define the unique extension which vanishes unless all of its arguments come
from $\T$. Doing this, $\alpha$ and $\beta$ remain cocycles, even though the
Lie bracket (and thus $d$) has changed. Moreover, they remain nontrivial. All
of this is contained in the following proposition:

\begin{prop} Let $\g$ and $\h$ be Lie superalgebras such that $\g$ acts on
	$\h$, and let $R$ be a representation of $\g \ltimes \h$. Given any
	$R$-valued $n$-cochain $\omega$ on $\h$, we can uniquely extend it to
	an $n$-cochain $\tilde{\omega}$ on $\g \ltimes \h$ that takes the value
	of $\omega$ on $\h$ and vanishes on $\g$. When $\omega$ is even, we
	have:
	\begin{enumerate}
		\item $\tilde{\omega}$ is closed if and only if $\omega$ is
			closed and $\g$-equivariant.
		\item $\tilde{\omega}$ is exact if and only if $\omega =
			d\theta$, for $\theta$ a $\g$-equivariant
			$(n-1)$-cochain on $\h$.
	\end{enumerate}
\end{prop}
\begin{proof}
	As a vector space, $\g \ltimes \h = \g \oplus \h$, so that 
	\[ \Lambda^n(\g \ltimes \h) \iso \bigoplus_{p+q=n} \Lambda^p \g \tensor \Lambda^q \h, \]
	as a vector space. Thanks to this decomposition, we can uniquely
	decompose $n$-cochains on $\g \ltimes \h$ by restricting to the summands.
	In keeping with our prior terminology, we call an $n$-cochain supported
	on $\Lambda^p \g \tensor \Lambda^q \h$ a $(p,q)$-form. Note that
	$\tilde{\omega}$ is just the $n$-cochain $\omega$ regarded as a
	$(0,n)$-form on $\g \ltimes \h$. We shall denote the space of
	$(p,q)$-forms by $C^{p,q}$. 
	
	We have two actions to distinguish: the action of $\g \ltimes \h$ on
	$R$, which we denote by $\rho$, and the action of $\g$ on $\h$, which
	we shall denote simply by the bracket, $[-,-]$. Inspecting the formula
	for the differential:
	\begin{eqnarray*}
		& & d\tilde{\omega}(X_1, \dots, X_{n+1}) = \\
		& & \sum^{n+1}_{i=1} (-1)^{i+1} (-1)^{|X_i||\tilde{\omega}|} \epsilon^{i-1}_1(i) \rho(X_i) \tilde{\omega}(X_1, \dots, \hat{X}_i, \dots, X_{n+1}) \\
		& & + \sum_{i < j} (-1)^{i+j} (-1)^{|X_i||X_j|} \epsilon^{i-1}_1(i) \epsilon^{j-1}_1(j) \tilde{\omega}([X_i, X_j], X_1, \dots, \hat{X}_i, \dots, \hat{X}_j, \dots X_{n+1})
	\end{eqnarray*}
	it is easy to see that
	\[ d \maps C^{p,q} \to C^{p,q+1} \oplus C^{p+1,q}. \]
	In particular:
	\[ d \maps C^{0,n} \to C^{0,n+1} \oplus C^{1,n}. \]
	Given an $n$-cochain $\omega$ on $\h$, it is easy to see that the
	part of $d\tilde{\omega}$ which lies in $C^{0,n+1}$ is just
	$\widetilde{d\omega}$, the extension of the $(n+1)$-cochain $d\omega$
	to $\g \ltimes \h$. 
	
	Let $e\omega$ denote the $(1,n)$-form part of $d\tilde{\omega}$.  To
	express this explicitly, choose $Y_1 \in \g$ and $X_2, \dots, X_{n+1}
	\in \h$. By definition $e\omega(Y_1, X_2, \dots, X_{n+1}) =
	d\tilde{\omega}(Y_1, X_2, \dots, X_{n+1})$, and inspecting the formula
	for the differential once more, we see this consists of only two
	nonzero terms:
	\begin{eqnarray*}
		e\omega(Y_1, X_2, \dots, X_{n+1}) 
		& = & (-1)^{|\tilde{\omega}||Y_1|} \rho(Y_1) \tilde{\omega}(X_2, \dots, X_{n+1}) \\
		&   & + \sum_{i=2}^{n+1} (-1)^{i+1} \epsilon_2^{i-1}(i) \tilde{\omega}([Y_1, X_i], X_2, \dots, \hat{X}_i, \dots, X_{n+1}) \\
		& = & (-1)^{|\omega||Y_1|} \rho(Y_1) \omega(X_2 \wedge \dots \wedge X_{n+1}) - \omega([Y_1, X_2 \wedge \dots \wedge X_{n+1}]) \\
	\end{eqnarray*}
	In particular, note that for even $\omega$,  $e\omega = 0$ if and only
	if $\omega$ is $\g$-equivariant.

	To summarize, for any $n$-cochain $\omega$, we have that
	\[ d\tilde{\omega} = \widetilde{d\omega} + e\omega, \]
	where the first $d$ is defined on $\g \ltimes \h$, while the second is
	only defined on $\h$. The proof of 1 is now immediate: for even
	$\omega$, $d\tilde{\omega} = 0$ if and only if $\widetilde{d\omega} =
	0$ and $e\omega = 0$, which happens if and only $d\omega = 0$ and
	$\omega$ is $\g$-equivariant.

	To prove 2, suppose $\omega$ is even. Assume $\tilde{\omega} = d\chi$,
	for some $(n-1)$-cochain $\chi$ on $\g \ltimes \h$. Because $d\chi$ is
	an even $(0,n)$-form, we may assume $\chi$ is an even $(0,n-1)$-form,
	as any other part of $\chi$ is closed and does not contribute to
	$d\chi$.  Thus $\chi$ is the extension of an even $(n-1)$-cochain
	$\theta$ on $\h$. By our prior formula, we have:
	\[ \tilde{\omega} = d\tilde{\theta} = \widetilde{d\theta} + e\theta \]
	The left-hand side is a $(0,n)$-form, and thus the $(1,n-1)$-form part
	of the right-hand side, $e\theta$, vanishes. Thus $\theta$ is
	$\g$-equivariant, and $\tilde{\omega} = \widetilde{d\theta}$, which
	implies $\omega = d\theta$.  On the other hand, if $\omega = d\theta$
	and $\theta$ is $\g$-equivariant, then $e\theta = 0$ and thus
	$\tilde{\omega} = d\tilde{\theta}$.
\end{proof}

Thus we can extend $\alpha$ and $\beta$ to nonexact cocycles on the Poincar\'e
Lie superalgebra. Thanks to Theorem~\ref{trivd}, we know that $\alpha$ lets us
extend $\siso(n+1,1)$ to a nonexact Lie 2-superalgebra:

\begin{thm} \label{thm:superstring}
	In dimensions 3, 4, 6 and 10, there exists a nonexact Lie 2-superalgebra formed
	by extending the Poincar\'e superalgebra $\siso(n+1,1)$ by the
	3-cocycle $\alpha$, which we call we the \define{superstring Lie
	2-superalgebra}, \define{\boldmath{$\superstring(n+1,1)$}}.
\end{thm}
\noindent
Likewise, in dimensions one higher, $\beta$ lets us extend $\siso(n+2,1)$ to a
nonexact Lie 3-superalgebra. In the 11-dimensional case, this coincides with
the Lie 3-superalgebra which Sati, Schreiber and Stasheff call $\sugra(10,1)$
\cite{SSS}, which is the Koszul dual of an algebra defined by D'Auria and
Fr\'e~\cite{DAuriaFre}.
\begin{thm} \label{thm:2brane}
	In dimensions 4, 5, 7 and 11, there exists a nonexact Lie 3-superalgebra formed
	by extending the Poincar\'e superalgebra $\siso(n+2,1)$ by the
	4-cocycle $\beta$, which we call the \define{2-brane Lie
	3-superalgebra}, \define{2-\boldmath{$\brane(n+2,1)$}}.
\end{thm}

\subsection*{Acknowledgements}

We thank Tevian Dray, Robert Helling, Corinne Manogue, Chris Rogers,
Hisham Sati, James Stasheff, and Riccardo Nicoletti for useful
conversations.  We especially thank Urs Schreiber for many discussions
of higher gauge theory and $L_\infty$-superalgebras.  
This work was supported by the FQXi grant RFP2-08-04.

\end{document}